\pdfoutput=1
\documentclass[aps, pra, superscriptaddress,reprint, nofootinbib]{revtex4-1}

\usepackage{amsmath,amssymb,amsthm,verbatim}
\usepackage{enumerate}
\usepackage[colorlinks, breaklinks]{hyperref}
\usepackage{xcolor}
\usepackage{bm}
\usepackage{graphicx}
\usepackage{longtable}
\usepackage{array}
\usepackage{soul}

\newtheorem{corollary}{Corollary}

\newtheorem{proposition}{Proposition}

\newtheorem{theorem}{Theorem}

\newtheorem{property}{Property}

\newtheorem*{remark}{Remark}

\newcommand{\bra}[1]{\langle #1|}
\newcommand{\ket}[1]{|#1\rangle}
\newcommand{\braket}[3]{\langle #1|#2|#3\rangle}

\newcommand{\op}[2]{|#1\rangle \langle #2|}
\newcommand{\1}{{\openone}}

\newcommand{\mbb}{\mathbb}
\newcommand{\mc}{\mathcal}

\newcommand{\tr}{\text{Tr}}

\newcommand{\ov}[1]{\overline{ #1}}

\definecolor{cool_green}{rgb}{0.0, 0.5, 0.0}

\newcommand{\madhav}[1]{{\color{black} #1}}

\begin{document}
\title{Simple Bounds for One-shot Pure-State Distillation in General Resource Theories}
\author{Madhav Krishnan Vijayan}
\affiliation{Centre for Quantum Software and Information, University of Technology Sydney, NSW 2007, Australia}
\author{Eric Chitambar}
\affiliation{Department of Electrical and Computer Engineering, Coordinated Science Laboratory,
University of Illinois at Urbana-Champaign, Urbana, IL 61801}
\author{Min-Hsiu Hsieh}
\affiliation{Centre for Quantum Software and Information, University of Technology Sydney, NSW 2007, Australia}
\date{\today}
\begin{abstract}
We present bounds for distilling many copies of a pure state from an arbitrary initial state in a general quantum resource theory.  Our bounds apply to operations that are able to generate no more than a $\delta$ amount of resource, where $\delta\geq 0$ is a given parameter.  To maximize applicability of our upper bound, we assume little structure on the set of free states under consideration besides a weak form of superadditivity of the function $G_{min}(\rho)$, which measures the overlap between $\rho$ and the set of free states.  Our bounds are given in terms of this function and the robustness of resource.  Known results in coherence and entanglement theory are reproduced in this more general framework. 
\end{abstract}

\maketitle

\section{Introduction}
In the development of quantum information theory, the operational approach has played a crucial role. It has enabled the description of abstract properties of quantum systems in terms of practical and well-defined information processing tasks. From the operational approach, the idea of certain states being a resource for information processing arose and lead to the development of quantum resource theories, most prominently the resource theory of entanglement \cite{horodecki_quantum_2009}. Motivated by the success of entanglement theory, several other quantum resource theories (QRTs) have been identified and studied such as those of coherence \cite{baumgratz_quantifying_2014, streltsov_colloquium:_2017, winter_operational_2016}, thermodynamics \cite{brandao_resource_2013, gour_resource_2015}, non-uniformity (purity) \cite{horodecki_local_2003, horodecki_reversible_2003, gour_resource_2015-1}, asymmetry \cite{gour_resource_2008, marvian_theory_2013} and non-stabilizer (``magic state'') quantum computation \cite{Veitch-2014a}, to name a few.  See \cite{coecke_mathematical_2016, chitambar_quantum_2018} for a wider review of QRTs. 

While what constitutes as a resource can vary widely between different theories, some common structure is shared among many QRTs.  Broadly speaking, a QRT divides states and operations in quantum mechanics into ones that an experimenter has access to freely and ones which are costly to use; in other words, into those which are free and those which are resourceful.  Where one draws this division usually depends on the particular experimental or physical constraints under consideration.  Studying this general structure independent of specific QRTs has allowed for a better understanding of certain quantum information quantities.  For example, Brand\~{a}o and Gour showed that the relative entropy of resource captures the asymptotic convertibility rate between two states, when one considers resource non-generating operations in a general convex QRT \cite{brandao_reversible_2015}. An operational interpretation for general resources was given in \cite{takagi_general_2019} by showing that for any convex QRT there exists a channel discrimination task for which a resource state will strictly outperform a free state.

In this work we present results that are most applicable to QRTs whose most resourceful states are pure.  This includes entanglement, coherence, and magic state quantum computing theories.  The meaning of ``most resourceful'' is ambiguous, yet it can be made more precise in both a quantitative and operational sense.  Quantitatively, pure states could be regarded as being more resourceful in a QRT if they maximize some resource measure, such as the relative entropy of resource \cite{Horodecki-2013a} or the robustness of resource \cite{brandao_reversible_2015}.  Alternatively, one could take an operational perspective and regard some set of pure states $S$ as being the most resourceful in a QRT if any state $\rho$ on a given state space can be realized by a free transformation $\varphi\to\rho$ with $\varphi:=\op{\varphi}{\varphi}\in\ S$.  In entanglement theory, such sets are known as maximally entangled sets, and it is an interesting research problem to identify maximally entangled sets with minimal structure \cite{deVicente-2013a}.  When pure states are regarded as a precious resource, a natural task of interest is pure-state distillation.  Typically, this problem is phrased as a multi-copy state conversion problem $\rho^{\otimes n}\to\varphi^{\otimes m}$, which can be interpreted as exchange $n$ copies of $\rho$ for $m$ copies of $\varphi$ using the free operations of the QRT.  In the limit of $n\to \infty$, the smallest ratio $\frac{n}{m}$ quantifies the asymptotic distillation rate of state $\varphi$ from $\rho$ \cite{Bennett-1996a}.  In the non-asymptotic or ``one-shot'' regime, the problem is to determine how many copies of $\varphi$ can be obtained from some initial state $\rho$ (possibly non-i.i.d.) up to some specified error bound \cite{liu_one-shot_2019, Barosz-2020a}.  The main contribution of this paper is to derive bounds on this one-shot pure-state distillation problem that apply to a wide-range of QRTs.

There is an intuitive link between the quantitative and operational resourcefulness just described.  The larger value a particular resource measure assigns to a state, the greater rate of resource distillation the state should possess.  While this rule of thumb does not always hold in general, usually it is possible to bound resource distillation rates in terms of other resource measures.  In this paper, we introduce a function $G_{min}(\rho)$ that measures the overlap of a state $\rho$ with the set of free states.  Our one-shot distillation bounds are given in terms of this function as well as the free robustness of the state.  The latter quantity measures how much mixing with another free state $\gamma$ is required to erase the resourcefulness of $\rho$.

Robustness is an important resource monotone first used to study entanglement \cite{vidal_robustness_1999}, and it has since found application in the study of general resource theories \cite{brandao_reversible_2015, takagi_general_2019, takagi_operational_2019}. Allowing the state $\gamma$ to be arbitrary and not necessarily free leads to the definition of the global robustness of resource. Every state will have a finite free robustness provided the set of free states has a non-empty interior.  However, for affine resource theories it can be shown that the free robustness will diverge for all resource states, and even for non-affine resource theories there can be states without finite free robustness \cite{liu_one-shot_2019, regula_convex_2017}.

To make our bounds applicable to more QRTs, we consider a smoothed version of the free robustness, which we will call the $\delta$-free robustness.  Roughly speaking, this quantity measures how much mixing of a free state is required to eliminate all but a $\delta\geq 0$ amount of resource from a given state.  In all QRTs, including affine ones, the $\delta$-free robustness will be finite for all states whenever $\delta > 0$. Complementing the $\delta$-free robustness is the set of quantum operations that cannot generate more than a $\delta$ amount of resource.  Our bounds are given for distilling pure states using these $\delta$-resource-generating operations.  Note that studying such operations has already proven crucial to obtaining asymptotic convertibility in entanglement theory \cite{Brandao-2008b, Brandao-2010b} and more general QRTs \cite{brandao_reversible_2015}.  To our knowledge, this is the first work that investigates a trade-off in resource-distillation with respect to a relaxation on the resource-generating power of the underlying operations.

During the completion of this work we became aware of an independent work  which derives bounds for the one-shot distillation rate in terms of the hypothesis testing relative entropy \cite{liu_one-shot_2019}.  We note that the hypothesis testing inequality is the operator smoothed version of $G_{min}(\rho)$ while we use the state smoothed version $G_{min}^{\epsilon}(\rho)$. Similarly the achievable map the authors in \cite{liu_one-shot_2019} use for mixed state transformation is a variation of the one we use for pure state distillation. The difference between these maps is that our map uses state smoothing instead of operator smoothing and is also applicable to QRTs where the free robustness need not be finite.  The authors define a class of QRTs in which there exists pure reference states that have constant overlap with the set free states which is conceptually similar to the constraints on $G_{min}(\phi^m)$ we introduce through Property \ref{Ax:extensive} as expressed in equation \eqref{Eq:Ax_main_alt}.  

\section{Definitions}
Let $\mc{S}$ denote the collection of all quantum states for a given quantum system.  A resource theory is defined by the pair $\{ \mc{F}, \mc{O} \}$ where $\mc{F}\subset \mc{S}$ is  called the set of free states and $\mc{O}$ is the set of free operations. Any state not in $\mc{F}$ is known as a resource state.  
One useful resource quantifier is the free robustness of resource, defined as
\begin{equation}
\label{Eq:defn-robustness}
\mc{R}_f(\rho):=\inf_{\pi\in\mc{F}}\left\{s\geq 0\;:\;\frac{\rho+s\pi}{1+s}\in\mc{F}\right\}.
\end{equation}
We refer to $\pi_{\rho}\in\mc{F}$ as an optimal state if it can be used to achieve the infimum value in the definition of $\mc{R}_f(\rho)$.  The quantity $\mc{R}_f(\rho)$ has a natural operational interpretation.  Suppose that an experimenter Alice has access to a resource state $\rho$ in her laboratory.  Additionally, Alice has the capability to prepare any free state $\pi\in\mc{F}$.  With probability $\frac{1}{1+s}$ Alice forwards the state $\rho$ to Bob, while with probability $\frac{s}{1+s}$ she prepares some free state $\pi$ and sends it to Bob instead.  Bob's description of the received state is thus $\frac{1}{1+s}(\rho+s\pi)$.  The free robustness of $\rho$, quantifies the threshold value such that for any $s<\mc{R}_{f}(\rho)$, Bob's received state will assuredly still possess resource.

One drawback of the free robustness is that it is not finite in many QRTs.  For example, in the resource theory of coherence, it is not possible to mix a resource state (i.e. a non-diagonal density matrix) with a free state (i.e. a diagonal density matrix) to obtain another free state.  An alternative notion of robustness that does not generally suffer from this problem involves taking the infimum in equation \eqref{Eq:defn-robustness} over all states $\mc{S}$ instead of over just the free states $\mc F$ \cite{harrow_robustness_2003, steiner_generalized_2003}.  The resulting quantity is known as the global robustness $\mc{R}_g(\rho)$, and it has emerged as an important resource measure since its dual characterization often leads to computationally-friendly resource witnesses \cite{Brandao-2005a, Piani-2016a, Napoli-2016a, regula_convex_2017}.  Here we introduce a family of robustness measures that generalizes the free robustness.  

For mathematical convenience, let us first recall the global log-robustness, which is given by
\begin{equation}
\mc{LR}_g(\rho):=\log\left[1+\mc{R}_g(\rho)\right].
\end{equation}
It is not difficult to show that this quantity is sub-additive, meaning that
\begin{equation}
\mc{LR}_g(\rho\otimes\sigma)\leq \mc{LR}_g(\rho)+\mc{LR}_g(\sigma).
\end{equation}

We next define the set of $\delta$-free states,
\begin{equation}
\label{Eq:-delta-free}
\mc{F}^\delta:=\{\rho\;:\;\mc{LR}_g(\rho)\leq \delta\},
\end{equation}
which from sub-additivity satisfies
\begin{equation}
\rho, \sigma\in\mc{F}^\delta \quad\Rightarrow\quad \rho\otimes \sigma\in\mc{F}^{2\delta}
\end{equation}
provided $\mc{F}$ is closed under tensor products.
Then for $\delta\in [0,+\infty]$, we define the $\delta$-free robustness as
\begin{align}
\label{Eq:defn-delta-robustness}
\mc{R}^\delta(\rho):=\inf_{\pi\in\mc{F}^\delta}\left\{s\geq 0\;:\;\frac{\rho+s\pi}{1+s}\in\mc{F}^\delta\right\},
\end{align}
from which we recover $\mc{R}_f(\rho)=\mc{R}^0(\rho)$.  We can likewise consider the $\delta$-free log-robustness,
\begin{equation}
\mc{LR}^\delta(\rho):=\log[1+\mc{R}_f^\delta(\rho)].
\end{equation}

It can be shown easily that the $\delta$-free robustness for any state $\rho \in \mc S$ is finite if $\delta>0$ and the global robustness is finite. Indeed, using convexity of the global robustness, we can see that for any state $\rho$ and free state $\gamma$,
\begin{align}
\label{eq:rob_convex}
R_g\left(\frac{\rho + s\gamma}{1 + s}\right) &\leq \frac{1}{1 + s}R_g\left(\rho\right) + \frac{s}{1 + s}R_g\left(\gamma\right) ,  \notag \\ 
&= \frac{1}{1 + s}R_g\left(\rho\right). 
\end{align}
Equation \eqref{eq:rob_convex} implies that the resource in any state $\rho$ as quantified by the global robustness, can be made arbitrarily small by mixing sufficiently with a free state provided the global robustness of $\rho$ is finite. In other words for any $\delta > 0$, there exists some finite positive number $s^*$ such that $\frac{1}{1 + s^*}(\rho + s^* \pi) \in \mc F^{\delta}.$

  For a given QRT $\{\mc{F},\mc{O}\}$, equation \eqref{Eq:-delta-free} provides a relaxation on the set of free states.  A corresponding relaxation can be made on the free operations.  Following the lead of Ref. \cite{Brandao-2008b, Brandao-2010b, brandao_reversible_2015}, we let $\mc{O}^\delta$ denote the set of $\delta$-resource-generating ($\delta$-RG) operations as the full collection of operations that act invariantly on $\mc{F}^\delta$; i.e. 
\begin{align}
\Lambda\in\mc{O}^\delta\quad\Leftrightarrow\quad\Lambda(\gamma)\in \mc{F}^\delta\qquad\forall\gamma\in\mc{F}^\delta.
\end{align}

We are interested in the problem of converting a given state $\rho$ to multiple copies of some pure state $\varphi$ using the $\delta$-resource-generating operations of the theory.  More precisely, for an initial state $\rho$ and a target state $\varphi=\op{\varphi}{\varphi}$, the one-shot distillation rate of conversion for parameters $\epsilon,\delta\geq 0$ is defined as
\begin{equation}
\begin{split}
\mc{D}^{\delta,\epsilon}(\rho,\varphi)
 :=  \max_{m \in \mathbb{N}} \left\lbrace  m : \sup_{\Lambda \in \mathcal{O}^\delta} F^2(\Lambda(\rho), \varphi^{\otimes m}) \geq 1- \epsilon \right\rbrace.
\end{split}
\end{equation}
Here, the fidelity between two states is given by
\begin{equation}
F(\rho, \sigma) := \tr \left( \sqrt{\sqrt{\sigma} \rho \sqrt{\sigma} } \right) =  \| \sqrt{\rho}\sqrt{\sigma} \|_1,
\end{equation}
which for a pure state $\sigma=\op{\varphi}{\varphi}$ has the form $F(\rho,\varphi)=\sqrt{\bra{\varphi}\rho\ket{\varphi}}$.  We will use the notation $\varphi^{\otimes m}$ and $\varphi^{m}$ interchangeably throughout this paper.

To obtain bounds on $\mc{D}^{\delta,\epsilon}(\rho,\varphi)$, we first define a quantity $G_{min}(\rho)$ as a measure of the maximum overlap between a positive operator $\rho$ and the set of free states $\mc{F}$,
\begin{align}\label{Eq:G_min_def}
G_{\min}(\rho) = \inf_{\gamma \in \mc{F}} \left\{ -\log \tr(\rho \gamma) \right\}.
\end{align}
We will want a smoothing of $G_{\min}(\rho)$, which is handled using a standard method \cite{Renner-2005a}.  Let us denote the $\epsilon$-ball around a state $\rho$ by
\begin{equation}
b(\rho, \epsilon) = \left\{ \mbb{I}  \geq \ov{\rho} \geq 0:  F(\ov{\rho}, \rho) \geq 1 - \epsilon \right\}.
\end{equation}
The pure state ball around a state $\rho$ is similarly given by
\begin{equation}
b_*(\rho, \epsilon) = \left\{ \ov\psi \in b(\rho, \epsilon) \text{ s.t. }\ov\psi \text{ is pure} \right\}.
\end{equation}
Then the state-smoothed version of $G_{min}(\rho)$ is defined as
\begin{align}
G_{min}^{\epsilon}(\rho) = \max_{\ov\rho \in b(\rho, \epsilon)}G_{min}(\ov\rho).
\end{align}
The pure state smoothed version $G^{\epsilon}_{min, *}(\rho)$ has a similar meaning except with the maximization taken over $b_*(\rho, \epsilon)$ instead $b(\rho, \epsilon)$.  If $b_{*}(\rho, \epsilon)$ is an empty set we define $G^{\epsilon}_{min, *}(\rho)=0$.

\section{General Distillation Bounds}
As described above, the essential ingredients to a resource theory are the sets of free states $\mc{F}$ and free operations $\mc{O}$.  Most QRTs will have additional structure on these objects, such as convexity or closure of $\mc{F}$ under partial trace.  We wish to bound $\mc{D}^{\delta,\epsilon}(\rho,\varphi)$ with as few assumptions on the QRT as possible.  For our upper bound, we only require that $G_{\min}$ is an extensive resource measure for pure states.  More precisely, we make the following singular assumption:
\begin{property}
\label{Ax:extensive}
For every pure state $\varphi$, there exists a constant $c(\varphi)$ such that
\begin{equation}
\label{Eq:G_min-extensive}
G_{\min}(\varphi^{\otimes m})=\inf_{\gamma\in\mc{F}}-\log\text{\upshape\tr}\left(\varphi^{\otimes m}\gamma\right)\geq m\cdot c(\varphi)
\end{equation}
for all $m\in\mbb{N}$.
\end{property}
\noindent
In thermodynamics, an extensive property is additive under the addition of more systems.  Equation \eqref{Eq:G_min-extensive} expresses this condition in a general QRT for the quantity $G_{\min}$ and multiple copies of a pure state. This extensive property holds for the QRTs for entanglement, coherence and purity. We now give our first result.

\begin{theorem}
\label{Th:converse}
Let $\epsilon,\delta\geq 0$ be arbitrary.  For any resource theory satisfying property \ref{Ax:extensive},
\begin{equation}\label{pure_state_result_eqn}
\frac{G^{2\sqrt{2\epsilon}}_{\min}(\rho)+\log(1+\delta)}{c(\varphi)} \geq \mc{D}^{\delta,\epsilon}(\rho,\varphi).
\end{equation}
Moreover, if $\rho$ is a pure state, this bound can be tightened to read
\begin{equation}\label{pure_state_result_eqn-pure}
\frac{G^{2\epsilon}_{\min,*}(\rho)+\log(1+\delta)}{c(\varphi)} \geq \mc{D}^{\delta,\epsilon}(\rho,\varphi).
\end{equation}
\end{theorem}
\begin{proof}
Let $m$ be the highest rate achievable with error $\epsilon$. This implies that there exists a $\delta$-RG operation $\Lambda \in \mc O^\delta$ such that $F^2(\Lambda(\rho), \varphi^m) \geq 1 - \epsilon$. Property \ref{Ax:extensive} can be equivalently stated as,
\begin{align}\label{Eq:Ax_main_alt}
\varphi^{ m} \gamma \varphi^{ m} \leq \frac{1}{2^{mc(\varphi)}} \varphi^{ m} \hspace{1cm} \forall \gamma \in \mc{F}. \end{align}
To see this, note that using the definition of $G_{min}$ and the statement of property \ref{Ax:extensive} we have,
\begin{align}
&\min_{\gamma \in \mc{F}}-\log(\tr(\varphi^{m}\gamma)) \geq mc(\varphi) ,\\
\implies &-\log(\tr(\varphi^{m}\gamma)) \geq m c(\varphi) \hspace{0.5cm} \forall \gamma \in \mc{F} ,\\
\implies& \tr(\varphi^{m}\gamma) \leq 2^{-mc(\varphi) } ,\\
\implies&  \braket{\varphi^m}{\gamma}{\varphi^m} \leq 2^{-mc(\varphi)} ,\\
\implies&  \braket{\varphi^m}{\gamma}{\varphi^m} \varphi^m \leq 2^{-mc(\varphi)}\varphi^m ,\\
\implies& \varphi^m \gamma \varphi^m \leq \frac{1}{2^{mc(\varphi)}}\varphi^m.
\end{align}
Starting from the final expression, all the steps can be reversed to obtain the initial expression hence proving the equivalence.

Since $\Lambda\in\mc{O}^\delta$, for every $\gamma\in\mc{F}$, there exists some $\pi\in\mc{F}$ and $\sigma\in\mc{S}$ such that $\Lambda(\gamma)=(1+\delta)\pi-\delta\sigma$.  Then from equation \eqref{Eq:Ax_main_alt}, it follows that
\begin{equation}
\varphi^m\Lambda(\gamma)\varphi^m\leq \frac{1+\delta}{2^{mc(\varphi)}}\varphi^m.
\end{equation}
Multiplying both sides of this by $\Lambda(\rho)$ and taking the trace yields
\begin{align}
\label{Eq:first_step}
\tr(\Lambda(\rho) \varphi^m \Lambda(\gamma) \varphi^m ) \leq  \frac{1+\delta}{2^{mc(\varphi)}} \tr(\Lambda(\rho) \varphi^m ) 
\leq \frac{1+\delta}{2^{mc(\varphi)}}.  
\end{align} 
Using the cyclic property of trace and denoting the dual map of $\Lambda$ as $\Lambda^*$ gives

\begin{align}
&mc(\varphi)-\log(1+\delta) \leq -\log\tr(\varphi^m\Lambda(\rho) \varphi^m \Lambda(\gamma)) ,\\
&= -\log\tr(\Lambda^*(\varphi^m\Lambda(\rho) \varphi^m) \gamma) ,\\
&= -\log \tr(Q\gamma) \leq -\log \tr(\sqrt{Q} \rho \sqrt{Q} \gamma),
\end{align}
where in the last inequality we use the fact that $\ov\rho := \sqrt{Q} \rho \sqrt{Q} \leq Q$. 
Since $\gamma$ is an arbitrary free state, we can say that
\begin{align}
mc(\varphi) &\leq  \min_{\gamma \in \mc{F}} \left\{ -\log \tr(\overline{\rho} \gamma) \right\}  = G_{min}(\ov{\rho})+\log(1+\delta). \label{Eq:trace_step} 
\end{align}

We will now show that $\overline{\rho} \in b(\rho, 2\sqrt{2\epsilon})$. Note that,
\begin{align}\label{tr_qpsi_eqn_step}
\tr(Q\rho) &= \tr(\varphi^m \Lambda(\rho) \varphi^m \Lambda(\rho))) = \braket{\varphi^m}{\Lambda(\rho )}{\varphi^m}^2 ,\\
&= \left( F^2(\Lambda(\rho) , \varphi^m ) \right)^2 \geq 1 - 2\epsilon.
\end{align}
where for the last inequality we use the fact that $F^2(\Lambda(\rho) , \varphi^m ) \geq 1 - \epsilon $. From the gentle measurement lemma \cite{winter_coding_1999} we know that,
$
\| \rho - \ov\rho \|_1 \leq 2\sqrt{2\epsilon} .
$
This implies that
\begin{align}
F^2(\rho, \ov\rho) \geq 1 - 2\sqrt{2\epsilon}
\end{align}
and $\ov\rho \in b(\rho, 2\sqrt{2\epsilon})$.
In equation \eqref{tr_qpsi_eqn_step}, replacing the mixed state $\rho$ with the pure state $\psi$ we have,
\begin{align}
\tr(Q\psi)  \geq 1 - 2\epsilon.
\end{align}
Note that $\ov\psi = \sqrt{Q}\psi \sqrt{Q}$, hence
\begin{equation}
\label{Eq:Fidelity-pure-approx}
\begin{split}
F(\psi, \overline{\psi}) &=  \bra{\psi}\sqrt{Q}\ket{\psi}   \geq  \bra{\psi} Q\ket{\psi}=   \tr(Q\psi)  \geq 1  - 2\epsilon .
\end{split}
\end{equation}
Hence $\ov{\psi} \in b_*(\psi, 2\epsilon )$ and we see that,
\begin{equation} \label{Eq:pure_state_converse}
mc(\varphi) \leq \max_{\ov{\psi} \in b_*(\psi, 2\epsilon)} G_{min}(\ov{\psi})+\log(1+\delta).
\end{equation}
\end{proof}

We next consider the achievability of pure-state distillation.  While Theorem \ref{Th:converse} holds for a wide class of QRTs, including ones that are non-convex, our lower bound on $\mc{D}^{\delta,\epsilon}(\rho,\varphi)$ applies only for convex QRTs whose free states are closed under tensor products.  Before stating this, we observe a property of the $\delta$-free log-robustness which holds in such QRTs.  Unlike the global log-robustness, $\mc{LR}^\delta$ does not appear to be sub-additive in general.  However, we can at least provide the following bound.
\begin{proposition}
\begin{align}
\label{Eq:pseudo-subadditive-log-robustness}
\mc{LR}^{m\delta}(\rho^{\otimes m})&\leq  \log[1+(1+2\mc{R}^\delta(\rho))^m]-1\notag\\
&\leq m \log[1+2\mc{R}^\delta(\rho)]
\end{align}
for every integer $m$ and $\delta>0$.  
\end{proposition}
\begin{proof}
Let $\rho=(1+s)\pi-s\sigma$, where $s=\mc{R}^\delta(\rho)$ and $\pi,\sigma\in\mc{F}^\delta$.  We can then write $\rho^{\otimes m}$ as a linear combination of operators belonging to $\mc{F}^{m\delta}$, each of which is an $m$-part tensor product of the $\pi$ and $\sigma$.  From the definition, $\mc{LR}^{m\delta}(\rho^{\otimes m})$ is no greater than the logarithm of the positive weight in this linear combination.  The positive weight can be written as 
\begin{align}
\frac{1}{2}\left[((1+s)+s)^{ m}+((1+s)-s)^{m}\right]=\frac{1}{2}(1+(1+2s)^m).\notag
\end{align}
Taking a logarithm establishes the first inequality in \eqref{Eq:pseudo-subadditive-log-robustness}, and the second follows by observing $1\leq (1+2\mc{R}^\delta(\rho))^m$.
\end{proof}

 We use this inequality to establish a lower bound $\mc{D}^{\sigma,\epsilon}(\rho,\varphi)$.
\begin{proposition}
\label{Prop:Direct}
Consider any QRT in which the set of free states $\mc{F}$ is convex.  For any $\delta,\epsilon\geq 0$, 
\begin{equation}
\label{Eq:prop1-bound}
\mc{D}^{\delta,\epsilon}(\rho,\varphi)\geq m
\end{equation}
for any positive integer $m$ satisfying 
\[G^{2\epsilon}_{\min,*}(\rho)  \geq m \log[1+2\mc{R}^{\delta/m}(\madhav{\varphi})].\]
\end{proposition}
\begin{remark}
For the special case that $\delta=0$, one can take $m=\left\lfloor\frac{G^{2\epsilon}_{\min,*}(\rho)}{\log[1+2\mc{R}^{0}(\madhav{\varphi})]}\right\rfloor$ so that
\end{remark}
\begin{equation}
\mc{D}^{0,\epsilon}(\rho,\varphi)\geq\left\lfloor\frac{G^{2\epsilon}_{\min,*}(\rho)}{\log[1+2\mc{R}^{0}(\madhav{\varphi})]}\right\rfloor.
\end{equation}

\begin{proof}
Let $m>0$ satisfy $G^{2\epsilon}_{\min,*}(\rho)  \geq m \log[1+2\mc{R}^{\delta/m}(\madhav{\varphi})]$.  We will follow a standard approach of introducing a simple measure-and-prepare map that does the job (see, for example, \cite{Rains-2001a}).  Consider the CPTP map
\begin{align}
\label{Eq:map-measure-prepare}
\Lambda(\omega) = \tr [(I-\ov\psi)\omega]\pi_{\varphi^m} + \tr[\ov\psi \omega] \varphi^m,
\end{align}
where $\pi_{\varphi^m}$ is an optimal state chosen in the definition of $\mc{R}^{\delta/m}(\varphi^m)$ and $\ov{\psi}$ is an optimal state chosen in the definition of $G^{2\epsilon}_{\min,*}(\rho)$.  
We first verify that
\begin{align}
F(\varphi^m, \Lambda(\rho)) &\geq F^2(\varphi^m, \Lambda(\rho)) = \tr(\varphi^m\Lambda(\rho)) ,\\ 
&\geq \tr[\ov{\psi}\rho] \geq 1-2\epsilon,
\end{align}
where we use the fact that $\ov{\psi}\in b_*(\rho,2\epsilon)$. Next, we use Eq. \eqref{Eq:pseudo-subadditive-log-robustness}
\begin{align}
\label{Eq:G-min-log-rob}
G^{2\epsilon}_{\min,*}(\rho)\geq m \log[1+2\mc{R}^{\delta/m}(\madhav{\varphi})]\geq \mc{LR}^\delta(\varphi^m),
\end{align}
along with the definition of $G^{2\epsilon}_{\min,*}(\rho)$ to conclude that
\begin{align}
\notag
-\log\tr[\ov\psi\gamma] \geq \mc{LR}_f^\delta(\varphi^m) 
\implies  \tr[\ov\psi\gamma]\leq [1+\mc{R}^\delta(\varphi^m)]^{-1}
\end{align}
for any $\gamma\in\mc{F}$.  Hence
\begin{align}
\Lambda(\gamma)&=\tr [(I-\ov\psi)\gamma]\pi_{\varphi^m} + \tr[\ov\psi \gamma] \varphi^m\in\mc{F}^\delta.
\end{align}
Convexity of $\mc{F}$ has been used here to ensure that $\mc{R}^\delta(\varphi^m) [1+\mc{R}^\delta(\varphi^m)]^{-1}\pi_{\varphi^m} + [1+\mc{R}^\delta(\varphi^m)]^{-1} \varphi^m$ remains free under any mixing with $\pi_{\varphi^m}$.
\end{proof}

We remark that the lower bound in proposition \ref{Prop:Direct} would be tighter if we could replace the $\delta$-free robustness in equation \eqref{Eq:prop1-bound} with the global robustness.  However doing so would no longer ensure that the measure-and-prepare map of equation \eqref{Eq:map-measure-prepare} always generates a sufficiently small amount of resource.  This problem can be overcome in the many-copy setting where one can invoke the Generalized Quantum Stein's Lemma \cite{Brandao-2010a}, and this is essentially the high-level approach taken in Refs. \cite{Brandao-2008b,Brandao-2010b, brandao_reversible_2015} to obtain asymptotic reversibility of resource transformations.

\section{Examples}

In many resource theories there exists a maximally resourceful unit pure state $\varphi$, such as the Bell state $\ket{\varphi_e}:=\sqrt{1/2}(\ket{00}+\ket{11})$ for entanglement or the uniform superposition state $\ket{\varphi_c}:=\sqrt{1/2}(\ket{0}+\ket{1})$ for coherence.  The one-shot distillation rate of the resource is the optimal rate at which one can convert a single copy of a given state state into several copies of the maximally resourceful unit state under some error threshold. Equation \eqref{Eq:pure_state_converse} immediately recovers known results for the one-shot concentration rate in entanglement \cite{buscemi_general_2013} and coherence \cite{vijayan_one-shot_2018, regula_one-shot_2018} as we will show below.

Let us first recall that min-entropy of a state $\rho$ is defined as.
\begin{equation}
S_{\min}(\rho) = -\log(\lambda_{\max}(\rho)),
\end{equation}
where $\lambda_{max}(\rho)$ is the largest eigenvalue of $\rho$.  In the QRT of entanglement, $G_{\min}(\varphi^{AB})=S_{\min}(\tr_{A}\varphi^{AB})$ while in the QRT of coherence, $G_{\min}(\varphi^{A})=S_{\min}(\Delta(\varphi^{A}))$, where $\Delta$ is the completely dephasing map. To see this notice that for any bipartite pure state, 
\begin{equation}
\ket{\varphi}^{AB} = \sum\limits_i \sqrt{\lambda_i}U\ket{\lambda_i}^A \ket{\lambda_i}^B 
\end{equation}
the minimisation in equation \eqref{Eq:G_min_def} is achieved by the product state $\gamma = U \op{\lambda_{max}}{\lambda_{\max}} U^{\dagger}  \otimes \op{\lambda_{max}}{\lambda_{\max}} $, where $\ket{\lambda_{max}}$ is the eigenvector corresponding to the largest eigenvalue of $\tr_{A}\varphi^{AB}$. Similarly for coherence the minimisation is achieved by the largest eigenvector of $\Delta(\varphi)$. Let $\Gamma$ represent the partial trace operation or the completely dephasing map. It can be easily verified that $S_{min}(\Gamma(\varphi^m)) = mS_{min}(\Gamma(\varphi))$. Comparing with equation \eqref{Eq:G_min-extensive}, the quantity $c(\varphi) = S_{min}(\Gamma(\varphi)) = 1$ for these resource theories.  

We define the ideal rate of one-shot distillation of entanglement for an arbitrary pure state $\rho$ to many copies of the maximally entangled state $\varphi_e$ using $\delta$-entanglement-generating operations as $E^{\delta, \epsilon}(\rho, \varphi_e)$. Similarly the ideal rate of one-shot coherence distillation is defined to be $C^{\delta, \epsilon}(\rho, \varphi_c)$, where $\varphi_c$ is the maximally coherent state. 

\begin{corollary}
The one-shot pure state concentration rate for entanglement $E^{\delta, \epsilon}(\psi^{AB}, \varphi_e)$ using $\delta$-entanglement-generating operations and the one-shot concentration rate of coherence $C^{\delta, \epsilon}(\psi, \varphi_c)$ using $\delta$-coherence-generating operations are given by,
\begin{align}
&E^{\delta, \epsilon}(\psi^{AB}, \varphi_e) \leq \max_{\ov\psi^{AB}\in b_*( \psi^{AB} \notag , 2\epsilon)}S_{min}(\rho_{\ov\psi^{AB}}) \\ & \hspace{4.7cm}+ \log(1 + \delta),\\
&C^{\delta, \epsilon}(\psi, \varphi_c) \leq \max_{\ov\psi\in b_*( \psi, 2\epsilon)} S_{min}(\Delta(\psi)) + \log(1 + \delta),
\end{align}
respectively, where $\rho_{\ov\psi^{AB}} = \tr_A(\psi^{AB})$ is the reduced density matrix of $\ov\psi^{AB}$ and $\Delta(\psi) = \sum_i \op{i}{i} \psi \op{i}{i}$ is the completely dephased version of $\psi$ in the incoherent basis. 
\end{corollary}
In the limit of $\delta = 0$ we recover previously known results regarding the one-shot concentration of coherence and entanglement in \cite{buscemi_general_2013, vijayan_one-shot_2018, regula_one-shot_2018}. Proposition \ref{Prop:Direct} implies that for the resource theory of entanglement, the upper-bound given in Theorem \ref{Th:converse} is tight for prefect transformations recovering the known result in \cite{buscemi_general_2013} as shown below.

\begin{table}[]
\centering
\caption{Value of $G_{min}(\psi)$ in different theories}
\label{my-label}
\begin{tabular}{ccc}
\cline{1-3}
\multicolumn{1}{|c|}{R.Theory} & \multicolumn{1}{c|}{Entanglemnet} & \multicolumn{1}{c|}{Coherence}    \\ \cline{1-3}
\multicolumn{1}{|c|}{$G_{min}(\psi)$} & \multicolumn{1}{c|}{$S_{min}(\rho_{\psi})$} & \multicolumn{1}{c|}{$S_{min}(\Delta(\psi))$}   \\  \cline{1-3}
\end{tabular}
\end{table}

\begin{corollary}
For the resource theory of entanglement the perfect transformation $\psi \rightarrow \varphi^{ m}_e$, where $\varphi_e$ is the unit maximally entangled state is achievable with a free operation $\Lambda \in \mc{O}$  with a rate
\begin{align}
E^{0,0}(\psi^{AB}, \varphi_e) = G_{min}(\psi^{AB}) = S_{min}(\rho_{\psi^{AB}}),
\end{align}
where $\rho_{\psi^{AB}} = \tr_{B}(\psi^{AB})$.\end{corollary}
\begin{proof}
From proposition \ref{Prop:Direct} we know that there exists a free-operation $\Lambda$ in the limit $\delta , \epsilon \rightarrow 0$ which performs the transformation $\psi \rightarrow \varphi_e^{ m}$ if,
\begin{align}
G_{min}(\psi) &\geq \madhav{m}\log (1 + \madhav{2}\mc{R}_f(\varphi_e)) \geq \madhav{\mc{LR}(\varphi_e^m)} \notag \\ & \madhav{ = \mc{LR}_g(\varphi^m_e)} = m,
\label{Eq:direct_ent_proof}
\end{align}
where we have used \madhav{equation \eqref{Eq:G-min-log-rob}} and the fact that the free \madhav{log}-robustness of entanglement is equal to the global \madhav{log-}robustness of entanglement $\mc{LR}_g(\rho)$ for pure states and for the maximally entangled state of rank $2^m$ \madhav{the global  robustness} is given by \madhav{$m$}  \cite{harrow_robustness_2003, steiner_generalized_2003}. Combining equations \eqref{Eq:direct_ent_proof} and \eqref{Eq:pure_state_converse} in the limit $\delta, \epsilon \rightarrow 0$ gives the desired result.
\end{proof}

\begin{remark}
For any dimension $d \geq 2$, the $\delta$-free robustness of coherence $\mc{R}_{f}^{\delta}(\varphi^m)$ is achieved by the maximally mixed state $\mbb{I}_d = \frac{1}{d}\sum\limits_{i} \op{i}{i}$, where $m = \log d$.\end{remark}

\begin{proof}
Let the optimal incoherent state achieving $\mc{R}_{f}^{\delta}(\varphi^m)$ be $\pi_{\varphi^m}$. 
The twirling operation $T$ for a state $\rho$ is defined as an empirical average over all possible incoherent basis permutations of $\rho$.
Notice that for the maximally coherent state  $T(\varphi^m) = \varphi^m$. From the definition of $\delta$-resource robustness we have,
\begin{align}\label{Eq:before_twirl}
\rho = \frac{\varphi^m + \mc{R}^{\delta}_f(\varphi^m) \pi_{\varphi^m}}{1 + \mc{R}^{\delta}_{f}(\varphi^m)} \in \mc{I}^{\delta},
\end{align}
where $\mc{I}^{\delta}$ is the set of $\delta$-incoherent states. Applying the twirling operation on both sides of equation \eqref{Eq:before_twirl} we have,
\begin{align}\label{Eq:twirl}
T(\rho) = \frac{\varphi^m + \mc{R}^{\delta}_f(\varphi^m) T(\pi_{\varphi^m})}{1 + \mc{R}^{\delta}_{f}(\varphi^m)} \in \mc{I}^{\delta} .
\end{align}
The last inclusion follows from the fact that coherence is invariant under the twirling operation. Equation \eqref{Eq:twirl} implies that if mixing $\mc{R}_{f}^{\delta}(\varphi^m)$ amount of $\pi_{\varphi^m}$ with $\varphi^m$ gives you a state in $\mc{I}^{\delta}$ then mixing $\mc{R}_{f}^{\delta}(\varphi^m)$ amount of $T(\pi_{\varphi^m})$ will also give you a state in $\mc{I}^{\delta}$. For any incoherent state $\gamma$, $T(\gamma)$ will be the completely mixed state $\frac{\mbb{I}_d}{d}$. We can see this by noticing that the state $T(\gamma)$ is permutation invariant by virtue  of the twirling operation and the only permutation invariant incoherent state is the maximally mixed state. 
\end{proof}

\section{Discussion}

As mentioned in the introduction, the general one-shot distillation problem has been studied in Ref. \cite{liu_one-shot_2019} in a more exhaustive manner.   However, the techniques we use are different, and a salient point of our work is the relative mathematical simplicity of our techniques and bounds. Currently our results are confined to pure state distillation and a future direction would be to see if our techniques can find the most general mixed state transformation bounds. 

Another open question from our work is whether our bounds reproduce the asymptotic results in \cite{brandao_reversible_2015} under the usual regularisation procedure. This requires further investigation of the asymptotic properties of $G_{min}(\rho)$ and $\mc{LR}^\delta(\rho)$. One technical challenge in this direction is that we are constrained to use the free robustness instead of the global robustness to ensure that our direct map is a free map.  To directly apply the Generalized Quantum Stein's Lemma of Ref. \cite{Brandao-2010a} for obtaining asymptotic results \cite{brandao_reversible_2015}, a connection needs to be made between the $\delta$-free robustness and the global robustness.  

It is also of interest to explore what the nature of trade off between error $\delta$ in the used operation and error $\epsilon$ in the final state is and if they have some operational interpretation. Clearly these quantities must be inversely related since increasing $\delta$ allows you to use a larger set of operations which can get you closer to the target state and thus reducing $\epsilon$. Operationally one can interpret a non-zero $\delta$ to represent the resource consumed to perform the given task. Seeing whether this allows us to define a new resource measure and more quantitative statements regarding specific QRTs are left to future work.

For our lower bounds we have assumed that the QRT must be convex, an improvement to these bounds would be to find a way to relax this constraint to include non-convex QRTs as well like we do in our upper bounds.

\bibliographystyle{apsrev4-1}

\end{document}